\documentclass[11pt]{article}
\usepackage[letterpaper, margin=1in]{geometry}
\usepackage[utf8]{inputenc}
\def\PrintMode{0}%

\usepackage{tikz}
\usetikzlibrary{positioning}
\usepackage{pgffor}
\usetikzlibrary{decorations.pathreplacing}
\usepackage{expl3}
\usepackage[T1]{fontenc}
\usepackage{graphicx}
\usepackage{float}
\usepackage[lined,ruled,linesnumbered]{algorithm2e}
\usepackage{xspace}
\usepackage[utf8]{inputenc}
\usepackage[noadjust]{cite}
\usepackage{amsmath, amssymb, amsfonts, amsthm}
\usepackage{enumerate}
\usepackage{adjustbox}
\usepackage{xcolor}
\usepackage[all]{xy}
\usepackage[linktocpage=true,pagebackref=true,colorlinks,linkcolor=magenta,citecolor=blue,bookmarks,bookmarksopen,bookmarksnumbered]{hyperref}
\usepackage{bm}
\usepackage{bbm}
\usepackage[normalem]{ulem}
\usepackage{paralist}
\usepackage{gensymb}
\usepackage{thmtools}
\usepackage{rotating}
\usepackage{mdframed}
\usepackage{multirow}
\usepackage{appendix}
\usepackage{tabularx}
\usepackage{verbatim}
\usepackage{mathrsfs}
\usepackage{indentfirst}
\usepackage{mleftright}
\usepackage[nottoc]{tocbibind}%
\usepackage{complexity}
\usepackage{tablefootnote}
\usepackage{epigraph}
\usepackage{autobreak}
\usepackage{todonotes}

\ifnum\PrintMode=1

\usepackage{savetrees}
\else
\fi

\graphicspath{{./figs/}}

\usepackage[capitalize]{cleveref}%
\newtheorem{theorem}{Theorem}[section]

\newtheorem{proposition}[theorem]{Proposition}
\newtheorem{corollary}[theorem]{Corollary}

\newtheorem{open}[theorem]{Open Problem}
\newtheorem{conjecture}[theorem]{Conjecture}

\theoremstyle{definition}
\newtheorem{definition}[theorem]{Definition}

\newtheorem{remark}[theorem]{Remark}
\newtheorem*{remark*}{Remark}

\renewcommand*\backref[1]{\ifx#1\relax \else (cit.~on p.~#1) \fi} %

\makeatletter
\def\moverlay{\mathpalette\mov@rlay}
\def\mov@rlay#1#2{\leavevmode\vtop{%
		\baselineskip\z@skip \lineskiplimit-\maxdimen
		\ialign{\hfil$\m@th#1##$\hfil\cr#2\crcr}}}
\newcommand{\charfusion}[3][\mathord]{
	#1{\ifx#1\mathop\vphantom{#2}\fi
		\mathpalette\mov@rlay{#2\cr#3}
	}
	\ifx#1\mathop\expandafter\displaylimits\fi}
\makeatother

\renewcommand{\poly}{\mathrm{poly}}

\newcommand{\one}{\mathbf 1}

\newlang{\MCSP}{MCSP}
\newlang{\MFSP}{MFSP}
\newlang{\MKtP}{MKtP}
\newlang{\MKTP}{MKTP}
\newlang{\itrMCSP}{itrMCSP}
\newlang{\itrMKTP}{itrMKTP}
\newlang{\itrMINKT}{itrMINKT}
\newlang{\MINKT}{MINKT}
\newlang{\MINK}{MINK}
\newlang{\MINcKT}{MINcKT}
\newlang{\CMD}{CMD}
\newlang{\DCMD}{DCMD}
\newlang{\CGL}{CGL}
\newlang{\PARITY}{PARITY}
\newlang{\Empty}{\textsc{Empty}}
\newlang{\Avoid}{\textsc{Avoid}}
\newlang{\Sparsification}{\textsc{Sparsification}}
\newlang{\HamEst}{\mathsf{HammingEst}}
\newlang{\HamHit}{\mathsf{HammingHit}}
\newlang{\CktEval}{\textsc{Circuit-Eval}}
\newlang{\Hard}{\textsc{Hard}}
\newlang{\cHard}{\textsc{cHard}}
\newlang{\CAPP}{CAPP}
\newlang{\GapUNSAT}{GapUNSAT}
\newlang{\OV}{OV}
\renewlang{\PCP}{PCP}
\newlang{\PCPP}{PCPP}
\newclass{\Avg}{Avg}
\newclass{\ZPEXP}{ZPEXP}
\newclass{\DLOGTIME}{DLOGTIME}
\newclass{\ALOGTIME}{ALOGTIME}
\newclass{\ATIME}{ATIME}%
\newclass{\SZKA}{SZKA}
\newclass{\Laconic}{Laconic\text{-}}
\newclass{\APEPP}{APEPP}
\newclass{\SAPEPP}{SAPEPP}

\newclass{\TFSigma}{TF\Sigma}
\newclass{\NTIMEGUESS}{NTIMEGUESS}

\newlang{\Formula}{Formula}
\newlang{\THR}{THR}
\newlang{\EMAJ}{EMAJ}
\newlang{\MAJ}{MAJ}
\newlang{\SYM}{SYM}
\newlang{\DOR}{DOR}
\newlang{\ETHR}{ETHR}
\newlang{\Midbit}{Midbit}
\newlang{\LCS}{LCS}
\newlang{\TAUT}{TAUT}
\newlang{\Poly}{\text{-}Poly}

\newcommand{\supp}{\operatorname{supp}}

\newcommand{\F}{\mathbb{F}}

\renewcommand{\R}{\mathbb{R}}

\DeclareMathOperator*{\Span}{\mathrm{Span}}

\renewcommand{\epsilon}{\varepsilon}
\newcommand{\eps}{\epsilon}

\usepackage[draft, inline,marginclue,index]{fixme}  %
\fxsetup{theme=color,mode=multiuser}

\definecolor{color1}{RGB}{46,134,193}

\definecolor{color7}{RGB}{128,0,128}
\FXRegisterAuthor{yeyuan}{ayeyuan}{\color{red}Yeyuan}

\definecolor{color3}{RGB}{255,128,0}
\FXRegisterAuthor{zihan}{azihan}{\color{blue}Zihan}

\definecolor{color5}{RGB}{128,128,128}
\FXRegisterAuthor{task}{atask}{\colorbox{color5}{\color{white}Task}}

\newif\ifmynotes
\mynotestrue

\title{Combinatorial Bounds for List Recovery via Discrete Brascamp--Lieb Inequalities\footnote{An extended abstract of this work appears in the proceedings of STOC 2026.}}
\date{}
\author{Joshua Brakensiek\thanks{University of California, Berkeley. \href{mailto:josh.brakensiek@berkeley.edu}{\texttt{josh.brakensiek@berkeley.edu}}} 
\and
Yeyuan Chen\thanks{Department of EECS, University of Michigan, Ann Arbor. \href{mailto:yeyuanch@umich.edu}{\texttt{yeyuanch@umich.edu}}}  
\and
Manik Dhar\thanks{Department of Mathematics, Massachusetts Institute of Technology. \href{mailto:dmanik@mit.edu}{\texttt{dmanik@mit.edu}}} 
\and
Zihan Zhang\thanks{Simons Institute for the Theory of Computing, UC Berkeley and Institute for Advanced Study, Princeton. \href{mailto:zzhsdj@foxmail.com}{\texttt{zzhsdj@foxmail.com}}}  
}

\begin{document}
\maketitle

\begin{abstract}
In coding theory, the problem of list recovery asks one to find all codewords $c$ of a given code $C$ which such that at least $1-\rho$ fraction of the symbols of $c$ lie in some predetermined set of $\ell$ symbols for each coordinate of the code. A key question is bounding the maximum possible list size $L$ of such codewords for the given code $C$.

In this paper, we give novel combinatorial bounds on the list recoverability of various families of linear and folded linear codes, including random linear codes, random Reed--Solomon codes, explicit folded Reed--Solomon codes, and explicit univariate multiplicity codes. Our main result is that in all of these settings, we show that for code of rate $R$, when $\rho = 1 - R - \epsilon$ approaches capacity, the list size $L$ is at most $(\ell/(R+\epsilon))^{O(1+R/\epsilon)}$. These results also apply in the average-radius regime. Our result resolves a long-standing open question on whether $L$ can be bounded by a polynomial in $\ell$. In the zero-error regime, our bound on $L$ perfectly matches known lower bounds.

The primary technique is a novel application of a discrete entropic Brascamp--Lieb inequality to the problem of list recovery, allowing us to relate the local structure of each coordinate with the global structure of the recovered list. As a result of independent interest, we show that a recent result by Chen and Zhang (STOC 2025) on the list decodability of folded Reed--Solomon codes can be generalized into a novel Brascamp--Lieb type inequality.
\end{abstract}

\section{Introduction}

Error-correcting codes ensure reliable communication over noisy channels by embedding messages $m \in \Sigma^k$ into redundant codewords $c \in \Sigma^n$.
At a high level, a fundamental goal in coding theory is to understand the optimal trade-off between the rate of a code and its resilience to errors.
As a basic example, a code $\mathcal{C} \subseteq \Sigma^n$ with rate $R = k/n$ and minimum relative Hamming distance $\delta$ can be uniquely decoded from up to a $\delta/2$ fraction of adversarial symbol errors. The breakthrough idea of \emph{list decoding}, introduced by Elias~\cite{elias} in the 1950s, further extended this capability, allowing decoding beyond the unique decoding radius by outputting a small list of candidate messages. Over the past several decades, list decoding has become a core topic in modern coding theory, with far-reaching influence across theoretical computer science --- including, but not limited to, the theory of  pseudorandomness \cite{sudan1999pseudorandom}, hard-core predicates \cite{GL89}, complexity theory \cite{CPS99}, and probabilistically checkable proofs \cite{arora1997improved}.

In this work, we focus on the notion of \emph{list recovery}, a powerful generalization of list decoding, in which the decoder must output a short list of candidate codewords consistent with multiple possible (potentially corrupted) symbols at each coordinate. Formally, for a code $\mathcal{C}\subseteq \Sigma^n$, we say $\mathcal{C}$ is $(\rho, \ell, L)$ list recoverable if, for any product set $S = S_1 \times \cdots \times S_n \in \binom{\Sigma}{\leq \ell}^n$ with $|S_1|,\dots,|S_n|\leq \ell$, there are at most $L$ codewords $c \in \mathcal{C}$ such that
\[\mathrm{dist}(c,S_1\times\cdots\times S_n)\leq\rho n,\]
where $\mathrm{dist}(c,S)$ denotes the number of indices $i\in[n]$ such that $c_i\notin S_i$. In particular, list decoding is the special case of list recovery when $\ell=1.$

During the past several decades, determining the optimal list recoverability of fundamental code families, such as linear codes, Reed–Solomon (RS) codes, and their variants, has remained a central challenge in coding theory. This challenge is not only of intrinsic interest, but also of great significance due to its profound impact on broader areas such as pseudorandomness~\cite{trevisan1999construction,ta2007lossless,guruswami2009Unbalanced,ta2012better,kalev2022unbalanced} and algorithm design~\cite{indyk2010efficiently,ngo2011efficiently,larsen2019heavy,dw22}.

Similarly to the list decoding case, the Johnson-type bound~\cite{gs01} guarantees that any maximum distance separable codes are $\left(1 - \sqrt{\ell R}, \ell, \mathrm{poly}(n)\right)$ list recoverable. Moreover, the notion of list recovery capacity has been formally established in~\cite{gi01, guruthesis, resthesis}. Specifically, in the large alphabet regime, when $q \ge 2^{\Omega(\log \ell / \epsilon)}$, plain random codes of rate $R$ and block length $n$ are $\left(\rho, \ell, O(\ell / \epsilon)\right)$ list recoverable with high probability, satisfying $\rho \ge 1 - R - \epsilon$ (see~\cite[Theorem 2.4.12]{resthesis} for more details). However, inspired by the work of Chen and Zhang~\cite{cz24}, it was recently shown that the output list size of random linear codes~\cite{lms25} and arbitrary linear codes~\cite{ls25} is lower bounded by $\ell^{\Omega(R / \epsilon)}$, thereby establishing a separation between plain random codes and linear codes in the context of list recovery.

Over the past decades, a long line of influential work \cite{ZP81,guruthesis,Guru06,guruswami2013linear,rw18,kopparty2018improved,kopparty2023improved,lp20,guo2022improved,goldberg2022list,tamo24,cz24,lms25,ls25} has advanced our understanding of the list recoverability of random linear codes, randomly punctured and explicit Reed--Solomon (RS) codes, as well as their variants, achieving parameters beyond the Johnson radius. More fine-grained introductions of these developments are listed below. For a more comprehensive history of list recovery, we refer the reader to the recent survey by Resch and Venkitesh~\cite{resch2025List}.
\begin{center}
\begin{table}
\footnotesize
\begin{tabular}{|m{7cm}|m{3cm}|m{3cm}|m{2cm}|m{2cm}|}
\hline 
Codes & Alphabet size $|\Sigma|$ & List size $L$ & Explicit?    \\ \hline

Random Linear Codes \cite{ZP81,guruthesis}&  $\ell^{O(1/\eps)}$& $|\Sigma|^{O(\ell/\eps)}$ & No\\ \hline
Random Linear Codes \cite{rw18} & $\ell^{O(1/\eps)}$ & $|\Sigma|^{O(\log^2(\ell)/\eps)}$ & No   \\\hline
Random Linear Codes \cite{ls25} & $\ell^{O(1/\eps)}$ & $(\ell/\eps)^{O(\ell/\eps)}$ & No\\
\hline
Random RS Codes \cite{ls25} & $O_{\eps,\ell,L,R}(n)$ & $(\ell/\eps)^{O(\ell/\eps)}$ & No\\
\hline
Random Linear Codes (This Work) &$\ell^{(\ell/(R+\eps))^{O(1+R/\eps)}/\eps}$ & $\left(\frac{\ell}{R+\epsilon}\right)^{O(1+R/\eps)}$&No \\ 
 \hline
Random RS Codes (This Work) & $O_{\eps,\ell,L,R}(n)$ & $\left(\frac{\ell}{R+\epsilon}\right)^{O(1+R/\eps)}$&No \\ 
 \hline
 \hline 

Folded RS Codes \cite{Guru06} & $\left( \frac{n}{\eps^2} \right)^{O(\log(\ell)/\eps^2)}$& $\left( \frac{n}{\eps^2} \right)^{O(\log(\ell)/\eps^2)}$ & Yes    \\ \hline 
Univariate Multiplicity Codes \cite{kopparty2015list} & $\left( \frac{n}{\eps^2} \right)^{O(\log(\ell)/\eps^2)}$& $\left( \frac{n}{\eps^2} \right)^{O(\log(\ell)/\eps^2)}$ & Yes   \\ \hline 
Folded RS/Univariate Multiplicity Codes \cite{guruswami2013linear} & $\left( \frac{n \ell}{\eps^2} \right)^{O(\ell/\eps^2)}$ & $\left( \frac{n \ell}{\eps} \right)^{O(\ell/\eps)}$ & Yes   \\ \hline
Folded RS/Univariate Multiplicity Codes \cite{kopparty2018improved,kopparty2023improved, tamo24} & $\left( \frac{ n \ell}{\eps^2} \right)^{O(\ell/\eps^2)}$& $(\ell/\eps)^{O((1+\log{\ell})/\eps)}$ & Yes     \\ \hline 
Folded AG Codes \cite{guruswami2013list, guruswami2016explicit} & $\exp\left( \frac{ \ell \log(\ell/\eps) }{\eps^2} \right)$ & $2^{2^{2^{O_{\eps,\ell}(\log^*(n))}}}$ & Yes  \\ \hline
Codes Based on the AEL Framework \cite{st25}& $\exp(\exp(\ell/\eps\log{(\ell/\eps)}))$ & $ \exp(\exp(\ell/\eps\log{(\ell/\eps)}))$ &Yes \\\hline

Codes Based on the AEL Framework \cite{JS25}& $\exp((\ell/\eps)^{(\ell/\eps)^{(\ell/\eps)}})$ & $(\ell/\eps)^{O(\ell/\eps)}$ &Yes \\\hline

Folded RS/Univariate Multiplicity Codes (This Work)& $\left( \frac{ n \ell}{\eps^2} \right)^{O(\ell/\eps^2)}$ &$\left(\frac{\ell}{R+\epsilon}\right)^{O(1+R/\eps)}$&Yes\\
 \hline
 \hline
 Any Linear Code \cite{guru21,cz24,lms25,ls25} & Any & $\ell^{\Omega(R/\eps)}$ & Lower Bound\\
 \hline
\end{tabular}
\caption{Adapted from \cite{kopparty2023improved}. Constructions of $(\alpha, \ell, L)$ list recoverable codes of rate $R=R^* - \eps$, where
$R^* = 1 - \alpha$ is list-recovering capacity (when $|\Sigma| \geq (1+\ell)^{\Omega(1/\epsilon)}$). We assume that $R^* \in (0,1)$ is constant (independent of $n, \eps, \ell$).} \label{tab:listrecovery}
\end{table}
\end{center}
\paragraph{List Recovery for Random Structured Codes.} It has long been of significant interest to study the list recoverability of random linear codes. Initially, the argument of Zyablov and Pinsker \cite{ZP81} showed that random
linear codes of rate $R$ over alphabet $\Sigma$ of size $\ell^{O(1/\epsilon)}$ are $(1-R-\epsilon,\ell,|\Sigma|^{O(\ell/\epsilon)})$ list recoverable (see \cite{guruthesis} for more details). Subsequent advances reduced the output list size to $|\Sigma|^{O(\log^2 \ell / \epsilon)}$ in the work of Rudra and Wootters~\cite{rw18}, and later to $(\ell / \epsilon)^{O(\ell / \epsilon)}$ in the recent work of Li and Shagrithaya~\cite{ls25}, which represents the current state of the art. 

Owing to a long line of influential works~\cite{rudra2014every,ST23,brakensiek2023generic,guo2023randomly,alrabiah2024randomly,Alrabiah2025Random}, the optimal list decodability of random RS codes is now fully understood. However, comparatively little is known about the list recoverability of random RS codes. Lund and Potukuchi~\cite{lp20} were the first to show that random RS codes are list recoverable beyond the Johnson bound. In particular, they proved that that random RS codes over the alphabet $\F_q$ of rate $O\left(\frac{1}{\sqrt{\ell}\log q}\right)$ are $\left(1-\frac{1}{\sqrt{2}},\ell,O(\ell)\right)$ list recoverable. This was later improved by Guo et al. \cite{guo2022improved}, who showed that the codes are $\left(1-\epsilon,\ell,O(\ell/\epsilon)\right)$ list recoverable with rate $O\left(\frac{\eps}{\sqrt{\ell}\log (1/\eps)}\right)$. Based on the reduction developed in~\cite{lms25}, the recent work of Li and Shagrithaya~\cite{ls25} also applies to random RS codes. Specifically, random RS codes of rate $R$ over alphabet of size $n\cdot(\ell/\eps)^{(\ell/\eps)^{O(\ell/\eps)}}$ are $\left(1-R-\epsilon,\ell,(\ell/\eps)^{O(\ell/\epsilon)}\right)$ list recoverable with high probability.

\paragraph{List Recovery for Explicit Structured Codes.} Explicit list recoverable codes such as folded RS \cite{krachkovsky2003reed,Guru06} and univariate multiplicity codes \cite{kopparty2015list} have been extensively studied over the past two decades. Guruswami and Rudra \cite{Guru06} first showed that folded RS codes of rate $R$ over alphabet of size $\left( \frac{n}{\eps^2} \right)^{O(\log(\ell)/\eps^2)}$ are $\left(1-R-\eps,\ell,\left( \frac{n}{\eps^2} \right)^{O(\log(\ell)/\eps^2)}\right)$ list recoverable. A similar result was extend to univariate multiplicity codes by the work of Kopartty \cite{kopparty2015list}. During the same period, a linear-algebraic decoding framework was promoted by Guruswami and Wang \cite{guruswami2013linear}, which implies that both folded RS and univariate multiplicity codes of rate $R$ achieve $\left(1-R-\eps,\ell,\left( \frac{n\ell}{\eps} \right)^{O(\ell/\eps)}\right)$ list recoverablility over alphabet of size $\left( \frac{ n \ell}{\eps^2} \right)^{O(\ell/\eps^2)}$. A subsequent line of work~\cite{kopparty2018improved,kopparty2023improved,tamo24} further improved the output list size from $\left(\tfrac{n\ell}{\epsilon}\right)^{O(\ell / \epsilon)}$ to $\left(\tfrac{\ell}{\epsilon}\right)^{O((1 + \log \ell) / \epsilon)}$, which remains the state of the art. 

For other structured codes, the work of \cite{guruswami2013list, guruswami2016explicit} shows that folded algebraic geometry (AG) codes of rate $R$ attain $\left(1-R-\eps,\ell,2^{2^{2^{O_{\eps,\ell}(\log^*(n))}}}\right)$ list recoverablility over alphabet of size $\exp\left( \frac{ \ell \log(\ell/\eps) }{\eps^2} \right)$. Most recently, Srivastava and Tulsiani~\cite{st25} showed that explicit codes based on the AEL framework~\cite{alon1995linear} with rate~$R$ can achieve $\left(1-R-\eps,\ell,\exp\left(\exp\left(\frac{\ell}{\eps}\log{(\frac{\ell}{\eps})}\right)\right)\right)$ over alphabet of size $\exp\left(\exp(\frac{\ell}{\eps}\log{(\frac{\ell}{\eps})})\right)$.

\paragraph{A Crucial Open Problem.} For codes of rate~$R$ that are $(1 - R - \epsilon, \ell, L)$ list recoverable, the best known bound~\cite{kopparty2018improved,kopparty2023improved,tamo24} on the output list size~$L$ for all the aforementioned structured codes over the past few decades is $\left(\tfrac{\ell}{\epsilon}\right)^{O((1 + \log \ell) / \epsilon)}$.
Meanwhile, by the most recent work~\cite{cz24,lms25,ls25}, the output list size of all these structured linear codes is lower bounded by $\ell^{\Omega(R / \epsilon)}$. Therefore, a substantial gap remains between the best known upper and lower bounds on the list recoverability of these codes. A crucial open problem, stated below and proposed multiple times in~\cite{rw18,kopparty2018improved,kopparty2023improved,ls25,resch2025List}, will be resolved in this work.
\begin{open}[See \cite{rw18,kopparty2018improved,kopparty2023improved,ls25,resch2025List}]\label{openp}

Can random linear codes, random RS codes, or explicit folded RS and univariate multiplicity codes of rate~$R$ achieve $(1-R-\epsilon, \ell, \poly_{\epsilon,R}(\ell))$ list recoverability? Moreover, can the output list size be further reduced from $\poly_{\epsilon,R}(\ell)$ to approach the lower bound $\ell^{\Omega(R / \epsilon)}$ as closely as possible?

\end{open}
\paragraph{A Conjecture of Chen and Zhang.} One potential approach to Open Problem~\ref{openp} is to determine the optimal relationship among the radius~$\rho$, rate~$R$, input list size~$\ell$, and output list size~$L$ in the context of list recovery for these structured codes. In particular, Chen and Zhang~\cite{cz24} proposed a conjecture concerning the optimal list recoverability of folded RS and univariate multiplicity codes. To formally state this conjecture for folded RS codes, we first recall their definition for convenience.

\begin{definition}[Folded Reed-Solomon code, quoting~\cite{cz24}] Given any parameters $s,n,k>0$, a finite field $\mathbb{F}_q$, where $|\mathbb{F}_q|>sn>k$, a generator $\gamma$ of $\mathbb{F}^\times_q$, and any sequence of $n$ elements $\alpha_1,\alpha_2,\dots,\alpha_n\in\F_q$, the corresponding $(s,\gamma)$-folded RS code over the alphabet $\mathbb{F}_q^s$ with block length $n$ and rate $R=\frac{k}{sn}$ is defined as 
\[\mathsf{FRS}^{s,\gamma}_{n,k}(\alpha_1,\alpha_2,\dots,\alpha_n):=\bigg\{\mathcal{C}(f):f\in\mathbb{F}_q[x],\text{ }\deg f<k\bigg\}\subseteq\left(\mathbb{F}_q^s\right)^n,
\]
where $\mathcal{C}(f):=(F_1,F_2,\dots,F_n)\in\left(\mathbb{F}_q^s\right)^n$, with $F_i:=\Bigl(f(\alpha_i),f(\gamma\alpha_i),\dots,f(\gamma^{s-1}\alpha_i)\Bigl)$, denotes the encoder of this code.
\end{definition}
The formal statement of the Chen--Zhang conjecture is stated below.
\begin{conjecture}[\text{See \cite[Conjecture 3.6]{cz24}}, Reformulated]\label{conj:main}
    For any constants $\epsilon>0, \ell\ge 2, L+1=\ell^a, a\in\mathbb{N}^{\ge 2}, \rho\leq 1-\frac{\ell}{L+1}$ and generator $\gamma$ of $\mathbb{F}^{\times}_q$ there exists a constant $C$ such that: if $s\ge C$, $\frac{k-1}{s}\ge a$, and $n$ is suffciently large, then for any rate 
$R=\frac{k}{sn}\leq R^*-\eps$,  folded RS code $\mathsf{FRS}^{s,\gamma}_{n,k}(\alpha_1,\alpha_2,\dots,\alpha_n)$ with appropriate evaluation points in $\mathbb{F}_q$ is $(\rho,\ell, L)$ list-recoverable, where
\[
R^*=\frac{a-1}{a}\left(1-\frac{\rho(L+1)}{L+1-\ell}\right).
\]
\end{conjecture}

In a recent work \cite{BCDZ25subspace}, it was shown that list-recoverability of random linear codes over large alphabets is essentially the same as the list-recoverability of nearly optimal subspace designable codes, so a proof or refutation of \cref{conj:main} that only relies on the ``subspace design'' nature of folded RS codes\footnote{We name a few (non)-examples here: The proof of \cite{cz24} only relies on the subspace design property, while \cite{guruswami2013linear,kopparty2023improved,tamo24,srivastava2025improved} use additional interpolation techniques that are tailored for folded RS codes.} is equivalent to resolving the same conjecture for random linear codes.

Beyond the trivial instances, our understanding of this conjecture in the genuine list-recovery regime (when $\ell \ge 2$) remains extremely limited. Prior to our work, the only two nontrivial pieces of evidence supporting~\cref{conj:main} came from~\cite[Theorem 1.8]{cz24} and~\cite{CCSZ25}, which established the conjecture for the specific parameters $\ell = 2, L = 3$ and $\ell = 2,\rho=0$ with arbitrary~$L$, respectively. 

We also give a formal definition of $s$-order multiplicity code.

\begin{definition}[Univariate multiplicity code, quoting \cite{cz24}]
Given parameters $s,n,k>0$, a finite prime field\footnote{For simplicity, we only consider univariate multiplicity codes defined over a finite prime 
field.} $\mathbb{F}_p$, where $|\mathbb{F}_p|\ge n$, and any sequence of $n$ elements  $\alpha_1,\dots,\alpha_n\in\mathbb{F}_p$, the corresponding $s$-order univariate multiplicity code over the alphabet $\mathbb{F}^s_p$ with block length $n$ and rate $R=\frac{k}{sn}$ is defined as 
\[\mathsf{MULT}_{n,k}^{(s)}(\alpha_1,\alpha_2,\dots,\alpha_n):=\bigg\{\mathcal{M}(f):f\in\mathbb{F}_p[x],\text{ }\deg f<k\bigg\}\subseteq\left(\mathbb{F}_p^s\right)^n,\]
where $\mathcal{M}(f):=(F_1,F_2,\dots,F_n)\in\left(\mathbb{F}_p^s\right)^n$, with $F_i:=\Bigl(f(\alpha_i),f^{(1)}(\alpha_i),\dots,f^{(s-1)}(\alpha_i)\Bigl)$, denotes the encoder of this code. Here $f^{(j)}(\alpha_i)$ denotes the $j$-th Hasse derivative.
\end{definition}

\subsection{Our Results}

Our main contribution in this paper is a positive resolution of \Cref{openp}. A more precise statement of our results is as follows.
\begin{theorem}[Main Results]\label{thm:main-lr}
For constants $\ell\ge 2,R,\eps\in(0,1),L=(\frac{\ell}{R+\eps})^{O(1+R/\eps)}$. The following codes are $(1-R-\eps,\ell,L)$ list-recoverable.
\begin{compactitem}
\item Explicit $s$-folded Reed--Solomon codes when $s\ge 16\ell/\eps^2$ \emph{(\cref{cor:frs-lr})}.
\item Explicit $s$-order univariate multiplicity codes when $s\ge 32\ell/\eps^2$ \emph{(\cref{cor:mult-lr})}.
\item Random $\F_q$-linear codes with high probability when $q\ge \exp(\log{(\ell)}L/\eps)$ \emph{(\cref{cor:rlc-lr})}.
\item Random Reed--Solomon codes over $\F_q$ with high probability when $q\ge cn$ for some constant $c=O_{\eps,\ell,L,R}(1)$ \emph{(\cref{cor:rrs-lr})}.
\end{compactitem}
\end{theorem}

See \Cref{tab:listrecovery} for a comparison of our results with prior work. The lower bound for output list size is $L\ge \ell^{\Omega(R/\eps)}$. Therefore, our bound $(\frac{\ell}{R+\eps})^{O(1+R/\eps)}$ is nearly optimal up to a constant multiplicative factor on the exponent when we consider the rate $R$ is a constant, which basically resolved Open Problem \ref{openp}. We also comment that the denominator $R+\eps$ is crucial when the rate $R$ approaches to $0$. It guarantees the denominator does not vanish as long as we have some  slack $\eps>0$ away from the threshold.

The above bound is nearly optimal asymptotically for radius $1-R-\eps$. However, if we fix $\ell,L\ge 2,\rho\leq 1-\frac{\ell}{L+1}$, it remains elusive to decide the exact rate $R$ such that linear codes could be $(\rho,\ell,L)$ list-recoverable. A conjecture on the exact value of $R$ is formulated as in \cref{conj:main}. 
In \cite{cz24}, it was proved that the rate $R$ must be upper bounded by the conjectured value $R^*$, so an ideal result is to prove $R\ge R^*-\eps$ for arbitrarily small slack $\eps>0$ and close the gap. In this work, we prove this for \textbf{zero-error} case.

\begin{theorem}[\cref{cor:clean-zero-bound}]\label{thm:zero-lr-2}
\emph{\cref{conj:main}} is true when $\rho=0$.
\end{theorem}
This  parameterized exact bound can be seen as the counterpart of the famous ``generalized singleton bound \cite{ST23}'' in zero-error list-recovery regime. We comment that the proof of \cref{thm:zero-lr-2} only uses the fact that explicit folded Reed--Solomon codes are nearly-optimal subspace designable, and there is a reduction \cref{thm:subspace-random} by \cite{BCDZ25subspace} from list-recoverability of optimal subspace designable codes to that of random linear codes, so a direct corollary from \cref{thm:zero-lr-2} and their reduction \cref{thm:subspace-random} is that \cref{conj:main} is also true for random linear codes when $\rho=0$.
\begin{remark}
When $\rho>0$, there is still a gap between our upper bound and the list size lower bound conjectured in \cref{conj:main}. We comment that it is impossible that our upper bound in this paper is optimal in general. To see why, \cite[Theorem 1.8]{cz24} has shown that these codes can be $(\frac{R}{2}-\eps,2,3)$ list-recoverable for arbitrarily small $\eps>0$, which proves \cref{conj:main} in a very special case. However, directly using proof techniques in this paper and plug in this set of parameter would only give worse bound, which is not optimal. We note that  \cite[Theorem 1.8]{cz24} is the only known special case of \cref{conj:main} that was proved but not follows from \cref{thm:zero-lr-2}. Although \cref{thm:main-lr} already gives nearly optimal bound for the radius $1-R-\eps$, we still think the fine-grained parameterized bound conjectured in \cref{conj:main} is a very interesting open problem.
\end{remark}

As we will discuss later, our main techniques is a discrete entropic Brascamp-Lieb inequality stated in \cref{thm:entropy-bl}. This inequality quantifies Shannon entropy.

Inspired by a recent result of Chen and Zhang~\cite{cz24} for list decoding, we prove in \cref{sec:rem-bl} a novel ``Brascamp--Lieb-like'' inequality, see \Cref{thm:rem-bl-intro}. However, instead of bounding the entropy of a probability distribution, we instead look at the ``remainder'' of a probability distribution $X\leftarrow \Omega$, which we define to be $r(X) :=1-\max_{u\in\Omega}(X(u))$. In \cref{sec:rem-inequality}, we show how this inequality can be used to recover Chen and Zhang's result that subspace designable codes asymptotically achieve the generalized Singleton bound~\cite{ST23}.

Finally, in \cref{app:BL} we give a self-contained proof of our main technique, discrete entropic Brascamp-Lieb inequality \cref{thm:entropy-bl}. Its original proof \cite{christ2013communication,carlen2009subadditivity} requires heavy hammers stemmed from pure mathematics community. We encourage readers to check our self-contained proof to get crucial ideas necessary in this discrete entropic form. We hope that could reveal some main intuitions behind this proof and make the result more accessible.

\subsection{Technical Overview}
The primary technical contribution of this paper is novel connection between discrete Brascamp--Lieb inequalities and list-recovery. The original continuous Brascamp--Lieb inequality~\cite{brascamp1976best} is used to prove various results in functional analysis (such as hypercontractivity) by relating some ``global'' continuous structure with many ``local'' projections of the same data. The functional analysis community has taken much interest in proving many generalizations of the original Brascamp--Lieb inequality~\cite{bennett2005finite,bennett2008brascamp,carlen2009subadditivity,bennett2024adjoint}, including in discrete settings~\cite{christ2013communication,christ2013optimal,bennett2024adjoint,christ2024multilinear}. Such inequalities have many applications in TCS including the study of matrix multiplication~\cite{christ2013communication}, non-commutative linear algebra \cite{GGOW20,franks2023shrunk}, quantum information theory~\cite{berta2023quantum} among other topics \cite{demmel2016parallelepipeds,bennett2022fourier}.

We crucially make use of a discrete Brascamp--Lieb inequality due to Christ, Demmel, Knight, Scanlon, and Yelick~\cite{christ2013communication,christ2024multilinear}. However, For convenience, we use a ``discrete entropic'' variant of their inequality based on a duality theorem of Carlen et al.~\cite{carlen2009subadditivity} (see also Remark 1.13 of \cite{bennett2024adjoint}). 

\begin{theorem}[Discrete Entropic Brascamp--Lieb~\cite{christ2013communication,christ2024multilinear,carlen2009subadditivity}]\label{thm:entropy-bl}
Let $V$ be a $d$-dimensional vector space over $\F_q$. For each $i \in [m]$, consider a scalar $s_i \ge 0$ and a linear map $\pi_i : V \to V_i$, where $V_i$ is a $d_i$-dimensional vector space over $\F_q$. Assume further that for all subspaces $U \subseteq V$ we have that
\begin{align}
\dim(U) \le \sum_{i=1}^m s_i \dim(\pi_i(U))\label{eq:dim-cond}
\end{align}
Then, for every probability distribution $X$ over $V$, we have that
\begin{align}
H(X) \le \sum_{i=1}^m s_i H(\pi_i(X)).\label{eq:entropy}
\end{align}
\end{theorem}
As the fundamental ideas behind the proof of Theorem~\ref{thm:entropy-bl} are rather elementary, we give a self-contained proof of this inequality in \Cref{app:BL}.

The novel step we take in this paper is to relate the linear maps $\pi_i : V \to V_i$ as encoding maps of a subspace design code. More precisely, we pick $V$ to be a suitable subcode of our subspace design code, and let each $\pi_i$ denote projection onto the $i$th coordinate. It turns out with an appropriate choice of the parameters $s_i \ge 0$, the equation (\ref{eq:dim-cond}) is \emph{equivalent} to a strong subspace design condition! As such, if the dimension of the subcode $V$ is bounded by the subspace design parameters, we get that (\ref{eq:entropy}) holds for \emph{any} probability distribution $X$ over $V$ for \emph{any} bounded-dimension subcode $V$ of \emph{any} subspace design code.

In our application to list recovery, we pick $X$ to (essentially) be the uniform distribution over a maximal recovered list (with $V$ being the span of the support of $X$). Then, each $H(\pi_i(X))$ can be bounded in terms of the local list size $\ell$ and the local error $\rho_i$--see Equation (\ref{eq:H-pi-bound}). Furthermore, since $X$ is a uniform distribution, $H(X)$ can be expressed in terms of the maximum list size $L$. As such, we near-immediately get an equation relating the global list size $L$ with local list size $\ell$ and the recovery error $\rho$. A careful asymptotic computation, yields our bound that $L$ is bounded by a polynomial function of $\ell$.

Since explicit folded Reed--Solomon codes and univariate multiplicity codes yield nearly optimal subspace designs \cite{guruswami2016explicit}, we immediately get sharp combinatorial list recovery bounds for these codes. In a recent work \cite{BCDZ25subspace}, the authors showed an equivalence between list-recoverability of subspace designs and random linear codes. Therefore, the list-recoverability of random linear codes directly follows from their reduction. Finally, the list-recoverability of random Reed--Solomon codes follows from a ``local equivalence'' between random RS and random linear codes established by \cite{lms25}.

\paragraph{Remainder BL Inequality.} To strengthen the connections between coding theory and Brascamp-Lieb inequalities, we aim to prove a novel Brascamp-Lieb-style inequality inspired by a recent result of Chen and Zhang~\cite{cz24} for list decoding. However, instead of bounding the entropy of a probability distribution, we look at a different quantity which we call \emph{remainder}. For any distribution $X\leftarrow \Omega$, let the remainder $r(X) :=1-\max_{u\in\Omega}(X(u))$. 

\begin{theorem}[Brascamp-Lieb inequality for remainder]\label{thm:rem-bl-intro}
Let $V$ be any linear space over $\F_q$. Let $\pi_i\colon V\to V/W_i,i\in[n]$ denote $n$ linear maps such that $\ker(\pi_i)=W_i\subseteq V$ satisfying the following
\begin{equation}\label{eq:subspace-design}
\forall W\subseteq V, \dim(W)\leq \sum^n_{i=1}s_i\dim(\pi_i(W)).
\end{equation}
Then, for any distribution $X\leftarrow V$, we know that
\[
r(X)\leq \sum^n_{i=1}s_ir(\pi_i(X)).
\]
\end{theorem}
We will see in \cref{sec:rem-inequality} how the above result gives a proof of \cite{cz24}. Note that although the entropic form \cref{thm:entropy-bl} proves \cref{thm:zero-lr-2}, it does not produce generalized singleton bound for list-decoding. Readers that are familiar with \cite{cz24} should think the original proof of \cite{cz24} as a special case of \cref{thm:rem-bl-intro} on uniform distributions only.

\paragraph{Streamlined Proof of \Cref{thm:entropy-bl}.} We conclude the technical overview by briefly describing an elementary proof we give of \Cref{thm:entropy-bl} in \Cref{app:BL}. It largely follows the proof strategy of \cite{christ2013communication,christ2024multilinear}, but replaces some (relatively) complex $L^p$ inequalities with some standard facts about entropy as well as  essentially eliminates the need to use facts about polyhedral geometry.

The proof proceeds by a double induction on $\dim(V)$ and $n$. The base cases are straightforward to handle. Observe that decreasing the parameters $s_i \ge 0$ only makes (\ref{eq:entropy}) stronger. Thus, without loss of generality we may always assume that (\ref{eq:dim-cond}) is an equality for some nonzero $U \subseteq V$. Such a space is called a \emph{critical} in the BL inequality literature (e.g., \cite{bennett2008brascamp,bennett2005finite}). The proof now splits into two cases: (A) $U \subsetneq V$ and (B) $U = V$.

Case (A) is handled using a divide-and-conquer, where we assume \Cref{thm:entropy-bl} holds for the spaces $U$ and $V/U$. Combining these smaller cases to yield (\ref{eq:entropy}) for every distribution over $V$ makes crucial use of the entropy chain rule. We remark this recursion on the critical subspace also appears in proofs of the GM-MDS theorem that the generator matrices of random Reed-Solomon codes attain all possible zero patterns~\cite{yildiz2019gmmds,lovett2018gmmds}. Case (B) is handled by perturbing the parameters $s_i \ge 0$ until some $U \subsetneq V$ is also critical, letting us reduce to case (A).

\subsection{Open Questions}

We conclude the introduction with a few open questions.

\begin{enumerate}
\item Our Theorem \ref{thm:zero-lr-2} establishes Conjecture \ref{conj:main} in the zero-error regime, and \cite[Theorem 1.8]{cz24} proves \cref{conj:main} when $\ell=2,L=3$.  However, beyond these regime and trivial cases, it remains open whether Conjecture \ref{conj:main} holds. In particular, we do not yet know in general the optimal trade-off among the decoding radius~$\rho$, rate~$R$, input list size~$\ell$, and output list size~$L$ for list recovery of such codes.
\item Recall that in \cite{st25}, their list-recovery result was not only combinatorial but also algorithmic. A natural question is whether our combinatorial guarantees in this work can likewise be made algorithmic.
\item \cref{cor:rlc-lr} requires the alphabet size to be at least $q= \ell^{\Omega(L/\eps)}$, where $L=(\ell/(R+\eps))^{\Theta(1+R/\eps)}$. This  lower bound is from a union bound over  local profiles required to characterize list-recoverability \cite[Proposition 2.2]{lms25}. Is it possible to prove similar results for smaller alphabets? By list-recovery capacity theorem \cite{resthesis}, all alphabets $q\ge \ell^{\Omega(1/\eps)}$ are possible to achieve $1-R-\eps$ radius.
\item We also leave open explicitly construct a code over alphabet size $q = \ell^{O(1/\eps)}$ which achieves our $L=(\ell/(R+\eps))^{O(1+R/\eps)}$ bound.
\end{enumerate}

\subsection*{Organization}

In \Cref{sec:prelim}, we outline our notation as well as facts about subspace designs and entropy. In \Cref{sec:zero-error}, we prove \Cref{cor:clean-zero-bound}. In \Cref{sec:asymptotic}, we prove \Cref{thm:main-lr}. In \Cref{sec:rem-bl}, we prove \Cref{thm:rem-bl-intro} and show how it connects to the list-decodability of codes. In \Cref{app:BL}, we give a streamlined proof of \Cref{thm:entropy-bl}.

\section{Preliminaries}\label{sec:prelim}
\paragraph{Notations.} For a code $\mathcal{C}\subseteq \Sigma^n$, we say $\mathcal{C}$ is $(\rho, \ell, L)$ list recoverable if, for any product set $S = S_1 \times \cdots \times S_n \in \binom{\Sigma}{\leq \ell}^n$ with $|S_1|,\dots,|S_n|\leq \ell$, there are at most $L$ codewords $c \in \mathcal{C}$ such that
\[\mathrm{dist}(c,S_1\times\cdots\times S_n)\leq\rho n,\]
where $\mathrm{dist}(c,S)$ denotes the number of indices $i\in[n]$ such that $c_i\notin S_i$. We use $\mathsf{LIST}_{\mathcal{C}}(\rho,S_1\times\cdots\times S_n)\subseteq \mathcal{C}$ to denote the set of codewords satisfying the above condition.
\subsection{Subspace Designs}

We now describe some standard notation for subspace designs (from, e.g., \cite{guruswami2016explicit,cz24}). Let $H_1, \hdots, H_m \subseteq \F_q^k$ be spaces with codimension $s$. We say that these spaces form an $(d, A)$- subspace design if for all $d$-dimensional subspaces $U \subseteq \F^k$ we have that
\[
    \sum_{i=1}^m \dim(U \cap H_i) \le A.
\]
A standard construction of a subspace design due to Guruswami and Kopparty~\cite{guruswami2016explicit} is to pick $sm$ distinct points $\alpha_1, \hdots, \alpha_{sm} \in \F_q$ such that for all $i \in [m]$,
\[
    H_i := \{x \in \F_q^k : \sum_{i=0}^{k-1} x_i \alpha_j^{i-1} = 0 \text{ for all } j \in \{s(i-1)+1, \hdots, si\}\}.
\]
They prove the following.
\begin{theorem}[\text{\cite[Theorem 7]{guruswami2016explicit}}]
For all $d \in \{0,1,\hdots, s\}$, $(H_1, \hdots, H_m)$ is a $\left(d, \frac{d(k-1)}{s-d+1}\right)$ subspace design.
\end{theorem}
\begin{remark}[Optimal parameters for subspace designs]\label{rem:optimal-sub}
Very recently, \cite[Theorem 6.1]{BCDZ25subspace} improved this bound to $\left(d, \frac{d(k-d)}{s-d+1}\right)$, and showed that this bound is tight for algebraically closed fields \cite[Theorem 6.2]{BCDZ25subspace}.
\end{remark}

\begin{definition}[\text{Subspace Designable Code, \cite[Definition B.2]{cz24}}]\label{def:designcodes} For any $s\ge 1$, given an $\mathbb{F}$-linear code $C\subseteq \left(\mathbb{F}^s\right)^n$ with message length $k$ and block length $n$, we use $\mathcal{C}\colon \mathbb{F}^k\to\left(\mathbb{F}^s\right)^n$ to denote the $\mathbb{F}$-linear encoder of $C$. For any $i\in[n]$, let $H_i\subseteq \mathbb{F}^k$ denote the $\mathbb{F}$-linear subspace such that for any message $f\in\mathbb{F}^k$, there is $\mathcal{C}(f)_i=0$ iff $f\in H_i$. We say $C$ is a $(\ell,A)$ subspace designable code if $\mathcal{H}:=\{H_1,\dots,H_n\}$ is an $(\ell, A)$ subspace design.
\end{definition}
\begin{theorem}[\cite{guruswami2016explicit}, as stated in \cite{cz24}]\label{thm:code-construction}For any $s,n\ge 1,q\ge sn\ge k$, there are explicit constructions of  $\F_q$-linear codes $\mathcal{C}\subseteq (\F^s_q)^n$ such that for all $d\in\{0,1,\dots,s\}$, $\mathcal{C}$ is a $(d,\frac{d(k-1)}{s-d+1})$ subspace designable code. In particular, $s$-folded explicit folded RS codes and $s$-order univariate multiplicity codes are such constructions
\end{theorem}
We call an $s$-folded $\F_q$-linear code $\mathcal{C}\subseteq (\F^s_q)^n$ with rate $R$ \textbf{$d$-subspace designable} iff for all $1\leq d'\leq d$, $\mathcal{C}$ is a $(d',Rd'n+1)$ subspace designable code (see \cite{BCDZ25subspace}). Asymptotically, these parameters match the lower bound for algebraically closed fields as discussed in \cref{rem:optimal-sub}, so we think $d$-subspace designable codes to be codes with optimal subspace designs up to dimension $d$.

We now quote a handful of results on the list recovery of various families of codes, which we use to prove \Cref{thm:main-lr} in \Cref{sec:asymptotic}. The first result is a recent result on reducing the list recoverability of random linear codes to the list recoverability of subspace designable codes.

\begin{theorem}[{\cite[Corollary 4.3]{BCDZ25subspace}}]\label{thm:subspace-random}
Fix constants $\ell,L,d\ge 1,\eps,\rho,R\in(0,1),q\ge(3\ell)^{2(L+1)/\eps}$. Suppose for all large enough $n$, all $d$-subspace designable codes with rate $R$ are $(\rho,\ell,L)$ list-recoverable, then random $\F_q$-linear codes with rate $R-(L^2+2L+2)/n-\eps$ are $(\rho,\ell,L)$ list-recoverable with probability at least $1-o_n(1)$.  
\end{theorem}

Next, we cite a theorem on relating the list recoverability of random linear codes to the list recoverability of random Reed-Solomon codes.

\begin{theorem}[{\cite[Corollary 3.11]{lms25}}, Reformulated]\label{thm:random-rs-reduction}
Fix constants $\rho\in[0,1],L\ge\ell\ge 2$. There exists a constant $c$ such that the following holds. For any constants $\eps,R\in(0,1)$ and prime power $q\ge 2^{c/\eps}n$, if random $\F_q$-linear code with rate $R$ is $(\rho,\ell,L)$ list-recoverable with probability  at least $1-o_n(1)$, then random RS code over $\F_q$ with rate $R-\eps$ is $(\rho,\ell,L)$ list-recoverable with probability at least $1-o_n(1)$. 
\end{theorem}

We also cite the following facts on the list recoverability of folded Reed-Solomon codes and univariate multiplicity codes.

\begin{theorem}[{\cite[Theorem 7]{guruswami2013linear}}]\label{thm:gw13-affine}
For any constant $R,\eps\in(0,1),\ell\ge2$ and $s\ge16\ell/\eps^2$, $s$-folded RS codes $\mathcal{C}\subseteq (\F^s_q)^n$ with rate $R=k/(sn)$ and appropriate evaluation points have the following property. For any received table $S_1,\times\cdots\times S_n\in\binom{\F^s_q}{\leq \ell}$, the set $\{f\in\F^k_q\colon \mathcal{C}(f)\in\mathsf{LIST}_{\mathcal{C}}(1-R-\eps,S_1\times\cdots\times S_n)\}$ is contained in some 
affine subspace of $\F^k_q$ with dimension $4\ell/\eps$.
\end{theorem}

\begin{theorem}[{\cite[Theorem 3.7]{kopparty2023improved}}, Reformulated]\label{thm:kr23-affine}
For any constant $R,\eps\in(0,1),\ell\ge2$ and $s\ge32\ell/\eps^2$, $s$-folded univariate multiplicity codes $\mathcal{C}\subseteq (\F^s_p)^n$ over prime field $p\ge sn$ with rate $R=k/(sn)$ and appropriate evaluation points have the following property. For any received table $S_1,\times\cdots\times S_n\in\binom{\F^s_p}{\leq \ell}$, the set $\{f\in\F^k_p\colon \mathcal{C}(f)\in\mathsf{LIST}_{\mathcal{C}}(1-R-\eps,S_1\times\cdots\times S_n)\}$ is contained in some 
affine subspace of $\F^k_p$ with dimension $8\ell/\eps$.
\end{theorem}
\subsection{Entropy}

Given a probability distribution $X$ over a ground set $V$, for each $v \in V$, we let $\Pr[X=v]$ denote the probability that $v$ is sampled from $X$. We define the entropy of $X$ to be
\[
    H(X) := \sum_{v \in V} \Pr[X = v]\log\left(\frac{1}{\Pr[X = v]}\right),
\]
where by standard convention we assume that $0 \log(1/0) = 0$.

Given a function $f : V \to W$, we define $f(X)$ to be the distribution such that for all $w \in W$,
\[
    \Pr[f(X) = w] = \sum_{\substack{v \in V\\f(v) = w}} \Pr[X = v].
\]
Observe that if $f$ is an injection, then $H(f(X)) = H(X)$. 

Given a distributions $X$ over $V$ and $Y$ over $W$ (potentially correlated), we define the joint entropy $H(X,Y)$ to be
\[
    H(X,Y) := \sum_{v \in V} \sum_{w \in W} \Pr[X = v \wedge Y = w]\log\left(\frac{1}{\Pr[X = v \wedge Y = w]}\right),
\]
Likewise, we define the conditional entropy $H(X \mid Y)$ to be 
\begin{align}
    H(X \mid Y) = H(X,Y) - H(Y) =  \sum_{v \in V} \sum_{w \in W} \Pr[X = v \wedge Y = w]\log\left(\frac{\Pr[Y=w]}{\Pr[X = v \wedge Y = w]}\right),\label{eq:entropy-1}
\end{align}
We also use that
\begin{align}
    H(X \mid Y) = \sum_{y \in W} \Pr[Y = w] H(X \mid Y = w).\label{eq:entropy-2}

\end{align}

\subsection{Applying Brascamp--Lieb to Subspace Designs}

We now prove an interface between subspace design conditions and the Brascamp-Lieb inequalities.

\begin{proposition}\label{prop:BL-subspace}
Let $\pi_1, \hdots, \pi_n : \F_q^k \to \F_q^s$ be linear maps, and let $H_i := \ker \pi_i$ for all $i \in [n]$. Assume that $H_1, \hdots, H_n$ form a $(d, \frac{d(k-1)}{s-d+1})$  subspace design for all $d \in \{0,1,\hdots, s\}$. Let $V \subseteq \F_q^k$ be a subspace with $\dim(V) \le s$. Then, we have that
\begin{align}
\dim(U) \le \frac{1}{n - \frac{k-1}{s-\dim(V)+1}}\sum_{i=1}^n \dim(\pi_i(U))\label{eq:dim-cond-sd}
\end{align}
holds for any subspace $U\subseteq V$. As such, for every probability distribution $X$ over $V$, we have that
\begin{align}
H(X) \le \frac{1}{n - \frac{k-1}{s-\dim(V)+1}}\sum_{i=1}^n H(\pi_i(X)).\label{eq:entropy-sd}
\end{align}
\end{proposition}
\begin{proof}
Observe from the subspace design condition that
\begin{align*}
\sum_{i=1}^n \dim(\pi_i(U)) &= n\dim(U) - \sum_{i=1}^n \dim(U \cap H_i)\\
    &\ge n\dim(U) - \dim(U) \frac{k-1}{s - \dim(U) + 1}\\
    &\ge \dim(U) \left(n - \frac{k-1}{s-\dim(V)+1}\right),
\end{align*}
from which (\ref{eq:dim-cond-sd}) follows. We have that (\ref{eq:entropy-sd}) immediately follows from (\ref{eq:dim-cond-sd}) and Theorem~\ref{thm:entropy-bl}.
\end{proof}

\section{Warm-up: Tight Analysis for Zero-error List Recovery}\label{sec:zero-error}

To demonstrate the strength of our method, we show in the zero-error regime our results tightly match the lower bound established by Chen and Zhang~\cite{cz24}, which gives an affirmative answer to \cref{conj:main} when the target radius $\rho=0$.

\begin{theorem}\label{thm:lr-zero-error}
For all $i \in [n]$, let $\pi_i : \F^k \to \F^s$ be a linear map whose kernel is $H_i$. Assume that $H_1, \hdots, H_n$ form a $\left(\ell, \frac{\ell(k-1)}{s-\ell+1}\right)$  subspace design for all $\ell \in \{0,1,\hdots, s\}$. Let $C :=\left \{(\pi_1(x), \hdots, \pi_n(x)) \mid x \in \F^k\right\}$. Then, for all $L \ge \ell \ge 2$, and integers $m \ge 1$, we have that $C$ is $(0, \ell, L)$-list-recoverable if 
\begin{align}
    \frac{k-1}{n(s-L+1)} < 1 - \frac{\log \ell}{\log (L+1)}.\label{eq:zero-error}
\end{align}
\end{theorem}

\begin{proof}
Assume for sake of contradiction that there exists a bad list $E \subseteq C$ of size $L+1$ such that for all $i \in [n]$, the set $\{c_i : c \in E\}$ has cardinality at most $\ell$. Let $F \subseteq \F^k$ be the preimage of $E$ with respect to $(\pi_1, \hdots, \pi_n)$. That is, $x \in F$ if and only if $(\pi_1(x), \hdots, \pi_n(x)) \in E$. Note that $|F| = |E| = L+1$.

By a suitable translation, we may assume that $0 \in F$ without changing the cardinality of any set $\pi_i(F)$. Let $V$ be the linear span of $F$. Note that $\dim(V) \le L$. By, \Cref{prop:BL-subspace}, we have that for any probability distribution $X$ over $V$, we have that
\[
H(X) \le \frac{1}{n - \frac{k-1}{s-L+1}}\sum_{i=1}^n H(\pi_i(X)).
\]
Let $X$ be the uniform distribution over $F$. Then, $H(X) = \log (L+1)$. Further, since $|\pi_i(F)| \le \ell$ for all $i \in [n]$, we have that $H(\pi_i(X)) \le \log \ell$ for all $i \in [n]$. Thus,
\[
    \log (L+1) \le \frac{n}{n - \frac{k-1}{s-L+1}} \log \ell.
\]
This contradicts (\ref{eq:zero-error}), so the bad list $E$ does not exist.
\end{proof}

\begin{corollary}\label{cor:clean-zero-bound}
For all $\eps > 0$, if $C$ as in \Cref{thm:lr-zero-error} has rate $R \in (0,1)$ and $n,k,s = \Omega_{\ell,\eps}(1)$ are sufficiently large, then $C$ is $(0, \ell, \ell^{1/(1-R-\eps)})$ list recoverable.
\end{corollary}

\begin{proof}
Let $L := \lfloor \ell^{1/(1-R-\eps)}\rfloor$ and assume that $s \ge (\frac{R}{\eps}+1)(L-1)$. Then, $\frac{Rs}{s-L+1} \leq R+\eps$. So, $\frac{k-1}{n(s-L+1)}<\frac{Rs}{s-L+1} \le R+\eps$. Further,
\[
    1 - \frac{\log \ell}{\log (L+1)} \ge 1 - \frac{\log \ell}{\log \ell^{1/(1-R-\eps)}} = R+\eps.
\]
We have that (\ref{eq:zero-error}) holds, so $C$ is $(0, \ell, \ell^{1/(1-R-\eps)})$ list recoverable.
\end{proof}

Fix any $\ell,a\ge 2$, for arbitrarily small slack $\eps>0$, we can choose any $R<\frac{a-1}{a}-\eps$ and invoke the above result, this yields $(0,\ell,\ell^{a}-1)$ list-recoverable. Since the rate $R$ could be arbitrarily close to $\frac{a-1}{a}$, this gives an affirmative answer to \cref{conj:main} in zero-error case $\rho=0$.

\section{Asymptotic Bound}\label{sec:asymptotic}

Fix any $d\ge 1,\mu\in(0,1)$, we call an $s$-folded $\F_q$-linear code $\mathcal{C}\subseteq (\F^s_q)^n$ with rate $R$ \textbf{$\mu$-slacked $d$-subspace designable} iff for any $1\leq d'\leq d$, $\mathcal{C}$ is $(d',(R+\mu)d'n)$ subspace-designable. Note that for any constants $\mu,R>0$, when $n$ approaches to infinity, $d$-subspace designable codes are automatically 
$\mu$-slacked $d$-subspace designable codes. We treat this definition to be ``nearly optimal'' subspace designable codes up to dimension $d$ and a $\mu$ slack from the optimal parameter. 
In the following calculations, we denote $W(z)$ as the Lambert $W$ functionthat $W(z)e^{W(z)}=z$.
\begin{theorem}\label{thm:list-recovery}
For any $\ell,m\ge2,L=\ell^m-1,s,d\ge 1,\mu,R,\rho\in(0,1)$ such that the following holds
\[
\rho<1-\frac{R'\ln{(L+1)}}{W(R’(L+1)\ln{(L+1)/\ell)}},\text{ where }R'=R+\mu
\]
Suppose $\rho\leq 1-\frac{\ell}{L+1}$\footnote{If $\rho>1-\frac{\ell}{L+1}$, then any code cannot be $(\rho,\ell,L)$ list-recoverable, so this condition does not weaken our result.}. Let $\mathcal{C}\subseteq(\F^s_q)^n$ be an $\F_q$-linear and $\mu$-slacked  $d$-subspace designable code with rate $R=k/(sn)$. If for any received table $S_1,\times\cdots\times S_n\in\binom{\F^s_q}{\leq \ell}$, the set $\{f\in\F^k_q\colon \mathcal{C}(f)\in\mathsf{LIST}_{\mathcal{C}}(\rho,S_1\times\cdots\times S_n)\}$ is contained in some 
affine subspace of $\F^k_q$ with dimension $d$, or $d\ge L$, then $\mathcal{C}$ is $(\rho,\ell,L)$ list-recoverable.
\end{theorem}
\begin{remark}
We can choose $\rho=0,L+1=\ell^{1/(1-R')}+1$, then \cref{thm:list-recovery} concludes that $\mu$-slacked $L$-subspace designable codes are $(0,\ell,\ell^{1/(1-R-\mu)})$ list-recoverable. From this perspective \cref{cor:clean-zero-bound} is a special case of \cref{thm:list-recovery}.
\end{remark}
\begin{proof}[Proof of \cref{thm:list-recovery}]
Suppose by contrapositive there exist a received table $S_1,\dots,S_n\in\binom{\F^s_q}{\leq \ell}$ and distinct $f_1,\dots,f_{L+1}\in\{f\in\F^k_q\colon \mathcal{C}(f)\in\mathsf{LIST}_{\mathcal{C}}(\rho,S_1\times\cdots\times S_n)\}$. Then, if we define $\rho_i=|\{j\in[L+1]\colon \mathcal{C}(f_j)[i]\in S_i\}|/(L+1)$ for each $i\in[n]$, there is $\frac{1}{n}\sum^n_{i=1}\rho_i\leq \rho$. Since $S_i$ is allowed to picked as a set with size $\ell$ and the condition is an upper bound on $\sum^n_{i=1}\rho_i$. Without loss of generality we can always choose some $S_1,\dots,S_n$ satisfying the above conditions such that $\rho_i\leq 1-\frac{\ell}{L+1}$ for each $i\in[n]$.  If $f_1\neq 0$, we subtract $f_1$ from each of $f_1,\dots,f_{L+1}$ and subtract $\mathcal{C}(f_1)[i]$ from each element of $S_i$ for all $i\in[n]$. One can check that after these subtractions the above conditions still hold, so without loss of generality, we can always assume $f_1=0$.

For any $i\in[n]$, let $H_i\subseteq \mathbb{F}_q^k$ denote the $\mathbb{F}_q$-linear subspace such that for any message $f\in\mathbb{F}_q^k$, there is $\mathcal{C}(f)_i=0$ iff $f\in H_i$. Let $V=\Span(f_1,\dots,f_{L+1})$, and $s_1,\dots,s_n=\frac{1}{(1-R-\mu)n}$. For each $i\in[n]$ let $\pi_i\colon V\to V/(V\cap H_i)$ denote the natural linear map from $V$ to its quotient  space. From the statement we know $\dim(V)\leq d$ Now we check that these choices satisfy (\ref{eq:dim-cond}) so we can invoke\footnote{In the original statement the entropy is defined using base-$2$ logarithm, but here we add a $\ln2$ factor to both sides of (\ref{eq:entropy}) and use natural logarithm with base $e$ in order to match the definition of Lambert $W$ function.} \cref{thm:entropy-bl}. For any linear subspace $U\subseteq V$, we aim to prove that
\[
\dim(U)\leq\sum_{i=1}^ns_i\dim(\pi_i(U)).
\]
The above inequality is equivalent to
\begin{align*}
&(1-R-\mu)n\dim(U)\leq \sum^n_{i=1}(\dim(U)-\dim(U\cap H_i))\\
\Leftrightarrow &\sum^n_{i=1}\dim(U\cap H_i)\leq (R+\mu)n\dim(U).
\end{align*}
This inequality follows since $\mathcal{C}$ is $\mu$-slacked $d$-subspace designable and $\dim(U)\leq \dim(V)\leq d$.

Let $X$ be the uniform distribution over $f_1,\dots,f_{L+1}\in V$, from \cref{thm:entropy-bl} we know that
\[
(1-R')n\ln{(L+1)}=(1-R-\mu)nH(X)\leq \sum^n_{i=1}H(\pi_i(X)).
\]
Fix any $i\in[n]$, we know that for $a\neq b\in[L+1]$, if $\mathcal{C}(f_a)[i]=\mathcal{C}(f_b)[i]$, then there must be $\mathcal{C}(f_a-f_b)[i]=0$ and therefore $f_a-f_b\in H_i\cap V$ which implies $\pi_i(f_a)=\pi_i(f_b)$. Let $E_i\in\{0,1\}$ denote the indicator event that $E_i=[\mathcal{C}(X)[i]\in T_i]$. Recall the definition of $\rho_i$ that $\Pr[E_i]=1-\rho_i$, let $h(\cdot)$ denote the binary entropy function, by chain rule (see (\ref{eq:entropy-1}) and (\ref{eq:entropy-2})) we know that
\begin{align}
H(\pi_i(X))&=H(\pi_i(X), E_i)\nonumber\\
&= H(E_i)+H(\pi_i(X)\mid E_i)\nonumber\\
&\leq h(\rho_i)+(1-\rho_i)\ln{\ell}+\rho_i\ln{(\rho_i(L+1))}\nonumber\\
&=(1-\rho_i)\ln{\ell}+\rho_i\ln{(L+1)}-(1-\rho_i)\ln{(1-\rho_i)}\label{eq:H-pi-bound}
\end{align}
Let $g(x)=(1-x)\ln{\ell}+x\ln{(L+1)}-(1-x)\ln{(1-x)}$, by direct calculation, we know that $g'(x)=\ln(\frac{L+1}{\ell})+1+\ln{(1-x)},  g''(x)=-\frac{1}{1-x}$. When $x\in[0,1-\frac{\ell}{L+1}]$, we know that $g'(x)>0$ and $g''(x)<0$. This implies $g(x)$ is concave and monotone increasing in this range. Therefore, recall that $\rho,\rho_1,\dots,\rho_n\in[0,1-\frac{\ell}{L+1}]$. Since $\sum^n_{i=1}\rho_i\leq \rho n$, we know that
\[
\sum^n_{i=1}H(\pi_i(X))\leq \sum^n_{i=1}g(\rho_i)\leq ng(\rho)=n((1-\rho)\ln{\ell}+\rho \ln{(L+1)}-(1-\rho)\ln{(1-\rho)}).
\]
Therefore, it follows that
\[
(1-R')\ln{(L+1)}\leq (1-\rho)\ln{\ell}+\rho\ln{(L+1)}-(1-\rho)\ln{(1-\rho)}.
\]
Rearrange this inequality, it follows that
\[
R'\ln{(L+1)}\ge(1-\rho)\ln{\frac{(1-\rho)(L+1)}{\ell}}
\]
Let $t=1-\rho\ge\frac{\ell}{L+1}$, it implies that
\[
R'\ln{(L+1)}\ge t\ln{t}+t\ln{\frac{(L+1)}{\ell}}.
\]
Let $h=\frac{t(L+1)}{\ell}\ge 1$, we know that
\[
\frac{R'(L+1)\ln{(L+1)}}{\ell}\ge h\ln{h}
\]
Let $C=\frac{R’(L+1)\ln{(L+1)}}{\ell}> 0$ and $h'=\exp(W(C))$. From the definition of Lambert W function we know $h'\ln{h'}=C>0$ so $h'>1$. Since $h\ln{h}$ is increasing with respect to $h\ge 1$, the above inequality implies $h\leq h'$. However, we can compute that
\[
h=\frac{(1-\rho)(L+1)}{\ell}>C/W(C)=h'
\]
This is a contradiction, so we complete the proof.
\end{proof}
\begin{theorem}\label{thm:asym-list-recovery}
For any $\ell\ge2,d\ge 1,\eps\in(0,1)$, let $L=\left(\frac{\ell}{R+\eps/2}\right)^{3+2R/\eps}$. Let $\mathcal{C}\subseteq(\F^s_q)^n$ be an $\F_q$-linear and $\eps/3$-slacked $d$-subspace designable code with rate $R=k/(sn)$. If for any received table $S_1,\times\cdots\times S_n\in\binom{\F^s_q}{\leq \ell}$, the set $\{f\in\F^k_q\colon \mathcal{C}(f)\in\mathsf{LIST}_{\mathcal{C}}(1-R-\eps,S_1\times\cdots\times S_n)\}$ is contained in some 
affine subspace of $\F^k_q$ with dimension $d$, or $d\ge L$, then $\mathcal{C}$ is $(1-R-\eps,\ell,L)$ list-recoverable.
\end{theorem}

\begin{proof}
We set $\mu=\eps/3$ and $R'=R+\mu$. It suffices to find $m\ge 2$ such that $\ell^m\leq L$ and $\rho=1-\frac{R'm\ln{\ell}}{W(R'm\ell^{m-1}\ln{\ell})}-\eps/6\ge 1-R-\eps=1-R'-2\eps/3$ and then invoke \cref{thm:list-recovery} to get the bound. Then we calculate a valid choice of $m$.

Let $K=R'+\eps/2$, then we need $K\ge \frac{R'm\ln{\ell}}{W(R'm\ell^{m-1}\ln{\ell})}$, so $W(R'm\ell^{m-1}\ln{\ell})\ge \frac{R'm\ln{\ell}}{K}$.

Since $W(z)e^{W(z)}=z$ and $\frac{R'm\ln{\ell}}{K}>0$. Since $xe^x$ is monotone increasing when $x>0$, this implies $\frac{R'm\ln{\ell}}{K}\exp(\frac{R'm\ln{\ell}}{K})\leq R'm\ell^{m-1}\ln{\ell}$. It follows that
\begin{align*}
\ell^{m-1}\ge \frac{1}{R'+\eps/2}\exp(\frac{R'm\ln{\ell}}{R'+\eps/2})\Rightarrow \ell^m\ge \frac{\ell}{R'+\eps/2}\exp(m\ln{\ell}(1-\frac{\eps/2}{R'+\eps/2}))
\end{align*}
Therefore,  it is necessary to have
\[
\ell^m\ge (\frac{\ell}{R'+\eps/2})^{(R'+\eps/2)/(\eps/2)}
\]
The above inequality is satisfied when we choose $m$ such that
\[
\ell^m\ge (\frac{\ell}{R+\eps/2})^{2+2R/\eps}
\]
It follows that there must be a valid choice of $m$ satisfying the above condition and $\ell^m\leq L$.
\end{proof}

We now use \Cref{thm:asym-list-recovery} to bound the list recovery of the families of codes mentioned in \Cref{thm:main-lr}. We start with random linear codes.

\begin{corollary}\label{cor:rlc-lr}
For constants $\ell\ge 2,R,\eps\in(0,1),L=(\frac{\ell}{R+\eps/2})^{5+4R/\eps},q\ge(3\ell)^{6(L+1)/\eps}$, random $\F_q$-linear codes with rate $R$ are $(1-R-\eps,\ell,L)$ list-recoverable with probability at least $1-o_n(1)$.
\end{corollary}
\begin{proof}
Let $d=L$, from \cref{thm:asym-list-recovery}, any $\eps/6$-slacked  $d$-subspace designable code with rate $R+\eps/2$ is $(1-R-\eps,\ell,L)$ list-recoverable.  Since $\eps>0$ is a constant, when $n$ approaches to infinity any $d$-subspace designable codes is a $\eps/6$-slacked  $d$-subspace designable code. Therefore, for large enough $n$, any $d$-subspace designable  code with rate $R+\eps/2$ is $(1-R-\eps,\ell,L)$ list-recoverable. Then, by \cref{thm:subspace-random}, for all large enough $n$, random $\F_q$-linear code codes with rate $R+\eps/2-\eps/3-(L^2+2L+2)/n=R+\eps/6-o_n(1)\ge R$ is $(1-R-\eps,\ell,L)$ list-recoverable with probability at least $1-o_n(1)$. 
\end{proof}

Next, we prove \Cref{thm:main-lr} for random Reed-Solomon codes.

\begin{corollary}\label{cor:rrs-lr}
For constants $\ell\ge 2,R,\eps\in(0,1),L=(\frac{\ell}{R+\eps/2})^{9+8R/\eps}$, there exists a constant $c$ such that for any prime power $q\ge cn$, random RS codes over $\F_q$ with rate $R$ are $(1-R-\eps,\ell,L)$ list-recoverable with probability at least $1-o_n(1)$.
\end{corollary}
\begin{proof}
By \cref{cor:rlc-lr}, there exists a constant $C$ such that for any prime power $q>C$, random $\F_q$-linear code with rate $R+\eps/2$ is $(1-R-\eps,\ell,L)$ list-recoverable. Then, by \cref{thm:random-rs-reduction}, there exists a constant $c$ such that for any prime power $q\ge cn$, random RS codes over $\F_q$ with rate $R$ are $(1-R-\eps,\ell,L)$ list-recoverable.
\end{proof}

We now prove \Cref{thm:main-lr} for folded Reed-Solomon codes.

\begin{corollary}\label{cor:frs-lr}
For constants $\ell\ge 2,R,\eps\in(0,1),L=(\frac{\ell}{R+\eps/2})^{3+2R/\eps},s\ge16\ell/\eps^2$, any $s$-folded RS code $\mathcal{C}\subseteq(\F^s_q)^n$ with rate $R=k/(sn)$ and appropriate evaluation points is $(1-R-\eps,\ell,L)$ list-recoverable.
\end{corollary}
\begin{proof}
Let $d=4\ell/\eps$, by \cref{thm:gw13-affine}, for any received table $S_1,\times\cdots\times S_n\in\binom{\F^s_q}{\leq \ell}$, the set $\{f\in\F^k_q\colon \mathcal{C}(f)\in\mathsf{LIST}_{\mathcal{C}}(1-R-\eps,S_1\times\cdots\times S_n)\}$ is contained in some 
affine subspace of $\F^k_q$ with dimension $d$. Now it suffices to check that $\mathcal{C}$ is a $\eps/3$-slacked $d$ subspace designable code so that we can apply \cref{thm:asym-list-recovery} and finish the proof. For any $1\leq d'\leq d$, we aim to prove that $\mathcal{C}$ is $(d',(R+\eps/3)d'n)$ subspace designable. By \cref{thm:code-construction}, $\mathcal{C}$ is $(d',\frac{Rsd'n}{s-d'+1})$ subspace designable. By direct calculation, it follows that $\frac{Rs}{s-d'+1}\leq R+\eps/3$ when $s\ge16\ell/\eps^2$ and we achieve the required parameter.
\end{proof}

Finally, we prove \Cref{thm:main-lr} for univariate multiplicity codes.

\begin{corollary}\label{cor:mult-lr}
For constants $\ell\ge 2,R,\eps\in(0,1),L=(\frac{\ell}{R+\eps/2})^{3+2R/\eps},s\ge32\ell/\eps^2$, any $s$-order univariate multiplicity code $\mathcal{C}\subseteq(\F^s_p)^n$ over prime field $p\ge sn$ with rate $R=k/(sn)$ and appropriate evaluation points is $(1-R-\eps,\ell,L)$ list-recoverable.
\end{corollary}
\begin{proof}
By replacing \cref{thm:gw13-affine} with \cref{thm:kr23-affine}, the proof is the same as the previous proof for \cref{cor:frs-lr}.
\end{proof}
\begin{remark}
All these results can be generalized to a slightly stronger notion of average-radius list-recovery. However, in the above folded RS codes (\cref{cor:frs-lr}) and univariate multiplicity codes (\cref{cor:mult-lr}) cases, in order to reduce the folding parameter $s$, we used two facts \cref{thm:gw13-affine} \cite{guruswami2013linear} and \cref{thm:kr23-affine} \cite{kopparty2023improved} that do not apply to average-radius list-recovery. Therefore, to ensure the same average-radius list-recoverability, we need a larger folding requirement  $s\ge L(3R/\eps+1)$ so that explicit folded RS and univariate multiplicity codes are $\eps/3$-slacked $L$-subspace designable codes guaranteed by \cref{thm:code-construction}.
\end{remark}

\section{Remainder Version of Brascamp--Lieb Inequality}\label{sec:rem-bl}

In this section, our target is to show the Brascamp--Lieb inequality for remainder and see how it derives the main result of \cite{cz24}.

\begin{theorem}[Restatement of \cref{thm:rem-bl-intro}]\label{thm:rem-bl}
Let $V$ be any linear space over $\F_q$. Let $\pi_i\colon V\to V/W_i,i\in[n]$ denote $n$ linear maps such that $\ker(\pi_i)=W_i\subseteq V$ satisfying the following
\begin{equation}\label{eq:subspace-design}
\forall W\subseteq V, \dim(W)\leq \sum^n_{i=1}s_i\dim(\pi_i(W)).
\end{equation}
Then, for any distribution $X\leftarrow V$, we know that
\[
r(X)\leq \sum^n_{i=1}s_ir(\pi_i(X)).
\]
\end{theorem}

\subsection{Application to List-Decoding}\label{sec:rem-inequality}

To gain some intuition for \Cref{thm:rem-bl}, we show that it implies as a corollary, a result by Chen and Zhang~\cite{cz24} on the list decodability of folded Reed--Solomon codes. The following theorem is equivalent to \cite[Theorem B.5]{cz24}, which is a more general form of its main result.
\begin{theorem}\label{thm:cz}
For all $i \in [n]$, let $\pi_i : \F^k \to \F^s$ be a linear map whose kernel is $H_i$. Assume that $H_1, \hdots, H_n$ form a $\left(\ell, \frac{\ell(k-1)}{s-\ell+1}\right)$  subspace design for all $\ell \in \{0,1,\hdots, s\}$. Let $C :=\left \{(\pi_1(x), \hdots, \pi_n(x)) \mid x \in \F^k\right\}$. Then, for all $L \leq s$, we have that $C$ is $(\rho, L)$-list-recoverable if 
\begin{align}
    \rho < \frac{L}{L+1}\left(1 - \frac{k-1}{n(s-L+1)}\right).\label{eq:GSB}
\end{align}
\end{theorem}

\begin{proof}
Assume for sake of contradiction that there exists a bad list $E \subseteq C$ of size $L+1$ and a point $y \in (\F^s)^n$ such that
\begin{align}
    \sum_{c \in E} \sum_{i=1}^n \one [c_i \neq y_i] \le \rho (L+1)n.\label{eq:bound}
\end{align}
For each $i \in [n],$ let $\rho_i := \frac{1}{L+1}\sum_{c \in E}\one [c_i \neq y_i].$ Let $F \subseteq \F^k$ be the preimage of $E$ with respect to $(\pi_1, \hdots, \pi_n)$. That is, $x \in F$ if and only if $(\pi_1(x), \hdots, \pi_n(x)) \in E$. Note that $|F| = |E| = L+1$.

By a suitable translation of $E$ and $F$, we may assume that $0 \in F$ such that (\ref{eq:bound}) still holds for a suitably translated $y$. Let $V$ be the linear span of $F$. Note that $\dim(V) \le L$. By \Cref{thm:rem-bl} and \cref{prop:BL-subspace}, we have that for any probability distribution $X$ over $V$, we have that
\[
r(X) \le \frac{1}{n - \frac{k-1}{s-L+1}}\sum_{i=1}^n r(\pi_i(X)).
\]
Let $X$ be the uniform distribution over $F$. Then, $r(X) = \frac{L}{L+1}$. Further, for all $i \in [n]$, we have that $r(\pi_i(F)) \le \rho_i$. Thus,
\[
    \frac{L}{L+1} \le \frac{1}{n - \frac{k-1}{s-L+1}} \sum_{i=1}^n \rho_i = \frac{1}{n - \frac{k-1}{s-L+1}} \cdot \frac{1}{L+1}\sum_{c \in E} \sum_{i=1}^n \one [c_i \neq x_i] \le \frac{\rho n}{n - \frac{k-1}{s-L+1}}.
\]
This contradicts (\ref{eq:GSB}), so the bad list $E$ does not exist.
\end{proof}

\subsection{Proof of \Cref{thm:rem-bl}}

The overall proof strategy is similar to the hypergraph approach of \cite{cz24} although we need extra bookkeeping to deal with the weighted distributions. We restate \Cref{thm:rem-bl} as follows. Given a function $F : V \to \R$ and a function $\pi_i : V \to V_i$, we let $F^{/\pi_i} : V_i \to \R$ denote the function
\[
  F^{/\pi_i}(v) := F(\pi^{-1}(v)) := \sum_{u \in \pi_i^{-1}(v)} F(u).
\]
It suffices to prove that if (\ref{eq:subspace-design}) holds, then for any function $F : V \to \R_{\ge 0}$, we have that
\begin{align}
\|F\|_1 - \|F\|_{\infty}\leq \sum^n_{i=1}s_i\left[\|F\|_1 - \|F^{/\pi_i}\|_{\infty}\right].\label{eq:rem-goal}
\end{align}

We prove this result by induction on the size of the support of $F$. Without loss of generality we can also assume that $V$ is spanned by the vectors in the support of $F$. The base case of $|\supp F| \le 1$ is trivial. Now, assume without loss of generality that $\|F\|_{\infty} = F(0)$. Pick a basis $v_1, \hdots, v_d$ of $V$ such that for all $i \in [d]$, $F(v_i)$ is maximized conditioned on $v_1, \hdots, v_{i-1}$. Let $v_0 = 0,$ and for all $i \in \{0, 1, \hdots, d\}$, let $U_i = \operatorname{span}\{v_0, \hdots, v_i\}$. We next decompose $F = F_0 + \cdots F_d$ such that for all $i \in \{0, 1, \hdots, d\}$,
\[
F_i(v) = \begin{cases}
F(v) & v \in U_i \setminus U_{i-1}\\
0 & \text{otherwise}.
\end{cases}
\]
Since (\ref{eq:subspace-design}) holds, for all $i \in \{0, 1, \hdots, d\}$, we have that
\begin{align}
\|F_i\|_1 - \|F_i\|_{\infty}\leq \sum^n_{j=1}s_j\left[\|F_i\|_1 - \|F^{/\pi_j}_i\|_{\infty}\right].\label{eq:Fi}
\end{align}
Next, we apply (\ref{eq:subspace-design}) for $W \in \{U_0, \hdots, U_d\}$ to get\footnote{We use the convention $F(v_{d+1}) = 0$.}
\begin{align*}
  \sum_{i=1}^d F(v_i) &= \sum_{i=1}^d (F(v_i)-F(v_{i+1})) \dim(U_i)\\
  &\le \sum_{i=1}^d (F(v_i)-F(v_{i+1})) \sum^n_{j=1}s_j\dim(\pi_j(U_i))\\
  &= \sum_{j=1}^n s_j \sum_{i=1}^d (F(v_i)-F(v_{i+1})) \dim(\pi_j(U_i)).\\
  &=\sum_{j=1}^n s_j \sum_{i=1}^d F(v_i)(\dim(\pi_j(U_i)) - \dim(\pi_j(U_{i-1})))\\
  &=\sum_{j=1}^n s_j \sum_{i=1}^d F(v_i)\one[v_i \not\in \ker \pi_j+U_{i-1}].
\end{align*}
Adding this inequality to (\ref{eq:Fi}) for all $i \in \{1, \hdots, d\}$, we get that
\begin{align*}
\|F\|_1 - \|F\|_{\infty} &= \sum_{i=1}^d(\|F_i\|_1 - \|F_i\|_{\infty})  + \sum_{i=1}^d F(v_i)\\
                         &\le \sum_{j=1}^n s_j \sum_{i=1}^d\left[\|F_i\|_1 - \|F_i^{/\pi_j}\|_{\infty} + F(v_i)\one[v_i \not\in \ker \pi_j+U_{i-1}].\right].
\end{align*}
Thus, it suffices to prove for all $j \in [n]$ that
\begin{align*}
  \sum_{i=1}^d\left[\|F_i\|_1 - \|F_i^{/\pi_j}\|_{\infty} + F(v_i)\one[v_i \not\in \ker \pi_j+U_{i-1}].\right] \le \|F\|_1 - \|F^{/\pi_j}\|_{\infty}.
\end{align*}
Rearranging terms, we get the somewhat simpler expression
\begin{align}
  \sum_{i=0}^d F(v_i)\one[v_i \not\in \ker \pi_j+U_{i-1}] \le \sum_{i=0}^d \|F_i^{/\pi_j}\|_{\infty} - \|F^{/\pi_j}\|_{\infty}.\label{eq:rem-goalX}
\end{align}
Let $u \in V$ be such that $F(u + \ker \pi_j) = \|F^{/\pi_j}\|_{\infty}$. Let $I \subseteq \{0,1, \hdots, d\}$ be the set of indices for which $(u + \ker \pi_j) \cap (U_i \setminus U_{i-1}) \neq \phi$ (note this means if $u=0$ then $0\in I$). In other words, $I$ is the leading columns of an (affine) echelon basis of $u + \ker \pi_j$. Therefore, $|I| = \dim \ker \pi_j + 1$. 

Then, we can see that
\[
  \|F^{/\pi_j}\|_{\infty} = F(u + \ker \pi_j) = \sum_{i \in I} F((u + \ker \pi_j) \cap (U_i \setminus U_{i-1})) = \sum_{i \in I} F_i(u + \ker \pi_j) \le \sum_{i \in I} \|F_i^{/\pi_j}\|_{\infty}
\]
Furthermore, if $i \not\in I$, then $F(v_i) \le \|F_i^{/\pi_j}\|_{\infty}$ is a trivial bound. Thus, to show (\ref{eq:rem-goalX}), it suffices to show
\[
  \sum_{\substack{i \in [d]\\v_i \not\in \ker\pi_j+U_{i-1}}} F(v_i) \le \sum_{\substack{i \in \{0\}\cup [d] \setminus I}} F(v_i)
\]
or equivalently
\begin{align}
   \sum_{i \in I} F(v_i) \le \sum_{\substack{i \in \{0\}\cup [d]\\v_i \in \ker \pi_j+U_{i-1}}} F(v_i).\label{eq:monotone}
\end{align}
Since $F(v_0) \ge \cdots \ge F(v_d)$, it suffices to prove (\ref{eq:monotone}) when $F(v_0) = \cdots = F(v_a) = 1$ and $F(v_{a+1}) = \cdots = F(v_d) = 0$ for some $a \in \{0,1,\hdots, d\}$.

In this setting, the LHS of \eqref{eq:monotone}
equals $\operatorname{affdim}(W \cap (u+\ker \pi_j))$ where  $W := \operatorname{affspan}\{v_0, \hdots, v_a\} = \operatorname{span}\{v_1, \hdots, v_j\}$. The RHS of \eqref{eq:monotone} is $1+\operatorname{dim}(W\cap (\ker \pi_j))$ (the extra one is because of $v_0=0$).

Hence, to prove \eqref{eq:monotone}, we want to show that $\operatorname{affdim}(W \cap (u+\ker \pi_j))$ is upper bounded by $1+\operatorname{dim}(W\cap (\ker \pi_j))$ which is immediate (one vector can be used to shift $W \cap (u+\ker \pi_j)$ to the origin and the resulting subspace is contained in $W\cap (\ker \pi_j)$). This complete the proof of \Cref{thm:rem-bl}.

\section{Self-contained Proof of Theorem~\ref{thm:entropy-bl}}\label{app:BL}

The majority of the machinery used by \cite{christ2013communication,carlen2009subadditivity} to prove Theorem~\ref{thm:entropy-bl} is unnecessary. In this section, we give a streamlined proof of Theorem~\ref{thm:entropy-bl}.

We prove Theorem~\ref{thm:entropy-bl} by a double induction on $d$ (the dimension of $V$) and $m$. The base case of $d=0$ is trivial, as a $0$-dimensional vector space is a singleton, and the entropy of any probability distribution over that vector space must be zero.

For $d=1$, note that each map $\pi_i : V \to V_i$ is either an injection or the zero map. In the former case, $H(\pi_i(X)) = H(X) = \dim(\pi_i(V)) H(X)$. While in the latter case, $H(\pi_i(X)) = 0 = \dim(\pi_i(V)) H(X)$. If we multiply (\ref{eq:dim-cond}) with $U=V$ by $H(X)$, we then get that
\[
    H(X) \le \sum_{i=1}^m s_i \dim(\pi_i(V)) H(X) = \sum_{i=1}^m s_i H(\pi_i(X)),
\]
as desired.

We now assume $d \ge 2$ and that Theorem~\ref{thm:entropy-bl} holds for all dimensions strictly less than $d$. For this fixed value of $d$, we perform an induction on $m$. We first handle the base case of $m=1$. In order for (\ref{eq:dim-cond}) for $U = \ker \pi_1$ to hold, we need that $\dim(\ker \pi_1) = 0$. Thus, $\pi_1$ is an injection. Further, for (\ref{eq:dim-cond}) to hold for $U = V$, we must have that $s_1 \ge 1$. Then, observe that $H(X) = H(\pi_1(X)) \le s_1 H(\pi_1(X))$, as desired.

Now assume that $m \ge 2$. To avoid degenerate situations, make two further assumptions.
\begin{itemize}
\item All scalars $s_i$ are positive; else we can reduce to a smaller $m$ by deleting $i$.
\item All maps $\pi_i$ are nonzero. If $\pi_i$ is the zero map, then $\dim(\pi_i(U)) = 0$ for all $U \subseteq V$ and $H(\pi_i(X)) = 0$ for any distribution $X$ over $V$, so we may delete $i$, reducing to a smaller $m$.
\end{itemize}
We say that (\ref{eq:dim-cond}) is \emph{critical} (see \cite{bennett2008brascamp,bennett2005finite}) at $U$ if the two sides of the inequality are equal. Note that (\ref{eq:dim-cond}) is always critical at $U = 0$. Consider the possibility that (\ref{eq:dim-cond}) is not critical for every nonzero $U \subseteq V$. In that case, there exist $s'_1, \hdots, s'_m \ge 0$ such that the following three properties hold.
\begin{enumerate}
\item $s'_i \le s_i$ for all $i \in [m]$.
\item (\ref{eq:dim-cond}) holds for $(s'_1, \hdots, s'_m)$.
\item (\ref{eq:dim-cond}) is critical $(s'_1, \hdots, s'_m)$ at some nonzero $U \subseteq V$.
\end{enumerate}
Since entropy is nonnegative, proving Theorem~\ref{thm:entropy-bl} for $(s'_1, \hdots, s'_m)$ implies Theorem~\ref{thm:entropy-bl} for $(s_1, \hdots, s_m)$. Therefore, we may without loss of generality assume that some nonzero $W \subseteq V$ is critical for $(s_1, \hdots, s_m)$. We next split into cases based on which subspace is critical.

\subsection{Nonzero $W \subsetneq V$ Is Tight}\label{subsec:W}

Note that (\ref{eq:dim-cond}) holds for all $U \subseteq W$. Since $\dim(W) < \dim(V)$, we have by the induction hypothesis that for all distributions $Y$ over $W$, we have that
\begin{align}
    H(Y) \le \sum_{i=1}^m s_i H(\pi_i(Y)).\label{eq:entropy-Y}
\end{align}
Let $V/W$ be the $\dim(V)-\dim(W)$-dimensional vector space $\{v + W : v \in V\}$. For every $i \in [m]$, we can define $\psi_i : V/W \to V_i/\pi_i(W)$ as
\[
    \psi_i(v + W) = \pi_i(v) + \pi_i(W).
\]
Every of $V/W$ is of the form $(U + W)/W$ for some $U \subseteq V$. We then calculate that
\begin{align*}
\sum_{i=1}^m s_i \dim(\psi_i((U+W)/W)) &= \sum_{i=1}^m s_i \left[\dim(\pi_i(U+W)/\pi_i(W))\right]\\
&= \sum_{i=1}^m s_i \dim(\pi_i(U+W)) - \sum_{i=1}^m s_i \dim(\pi_i(W))\\
&\ge \dim(U+W) - \dim(W)\\
&= \dim((U+W)/W),
\end{align*}
where the third line uses (\ref{eq:dim-cond}) and the fact that it is critical at $W$. Since $\dim(V) - \dim(W) < \dim(V)$, by the induction hypothesis, we have that for all distributions $Z$ over $V/W$, we have that
\begin{align}
    H(Z) \le \sum_{i=1}^m s_i H(\psi_i(Z)).\label{eq:entropy-Z}
\end{align}
We now show how (\ref{eq:entropy-Y}) and (\ref{eq:entropy-Z}) imply (\ref{eq:entropy}).

Given our distribution $X$ over $V$, we define $X_{/W}$ to be the distribution over $V/W$ defined to be
\[
    \Pr[X_{/W} = v + W] = \Pr[X \in v + W] = \sum_{w \in W} \Pr[X = v + w].
\]
For each $v + W \in V/W$, we further define $X_{v+W}$ to be the distribution over $W$ defined by
\[
    \Pr[X_{v+W} = w] = \frac{\sum_{w'\in W} \Pr[X = v+w']}{\Pr[X \in v+W]}
\]
if $\Pr[X \in v+W] > 0$. Otherwise, we define $X_{v+W}$ to be the uniform distribution over $W$. We now give a crucial entropy inequality

\begin{proposition}\label{prop:ineq-entropy}
For any linear map $\pi : V \to V'$ with corresponding map $\psi : V/W \to V'/\pi(W)$, we have that
\begin{align}
    H(\pi(X)) \ge H(\psi(X_{/W})) + \sum_{v + W \in V/W} \Pr[X \in v+W] H(\pi(X_{v+W})),\label{eq:chain}
\end{align}
with equality when $\pi$ is an injection.
\end{proposition}
\begin{proof}
Assume that $X$ and $X_{/W}$ are correlated so that if $X$ samples $v$, then $X_{/W}$ samples $v+W$. In that case, since $\psi(v+W) = \pi(v) + \pi(W)$, we have that $\pi(X)$ determines $\psi(X_{/W})$. That is, $H(\psi(X_{/W}) \mid \pi(X)) = 0$. Thus,
\begin{align*}
H(\pi(X)) - H(\psi(X_{/W}))&= H(\pi(X), \psi(X_{/W})) - H(\psi(X_{/W}) \mid \pi(X)) - H(\psi(X_{/W}))\\
&=  H(\pi(X), \psi(X_{/W}))- H(\psi(X_{/W}))\\
&= H(\pi(X) \mid \psi(X_{/W}))\\
&\ge H(\pi(X) \mid X_{/W})\\
&= \sum_{v+W \in V/W} \Pr[X_{/W} = v+W] H(\pi(X) \mid X_{/W} = v+W)\\
&= \sum_{v+W \in V/W} \Pr[X \in v+W] H(\pi(X) \mid X \in v+W)\\
&= \sum_{v+W \in V/W} \Pr[X \in v+W]H(\pi(X_{v+W})),
\end{align*}
where for the last line if $\Pr[X \in v + W] = 0$, then equality holds no matter what $X_{v+W}$ is defined to be. Note that if $\pi$ is an injection, then $\psi$ is an injection so $H(\pi(X) \mid \psi(X_{/W})) = H(\pi(X) \mid X_{/W})$, rendering (\ref{eq:chain}) an equality.
\end{proof}

To finish this case, we observe by Proposition~\ref{prop:ineq-entropy} that
\begin{align*}
\sum_{i=1}^m s_i H(\pi_i(X)) &\ge \sum_{i=1}^m s_i \left[H(\psi_i(X_{/W})) + \sum_{v + W \in V/W} \Pr[X \in v+W] H(\pi_i(X_{v+W}))\right]\\
&= \sum_{i=1}^m s_i H(\psi_i(X_{/W})) +\sum_{v + W \in V/W} \Pr[X \in v+W] \sum_{i=1}^m s_i H(\pi_i(X_{v+W}))\\
&\ge H(X_{/W}) + \sum_{v + W \in V/W} \Pr[X \in v+W] H(X_{v+W})\\
&= H(X),
\end{align*}
where the third line uses (\ref{eq:entropy-Y}) and (\ref{eq:entropy-Z}), and the last line used the equality case of \Cref{prop:ineq-entropy} when $\pi : V \to V$ is the identity map.

\subsection{Only $V$ Is Tight} Here, we reduce to a case in which some other $W \subsetneq V$ is also critical. Let $d_{m-1} = \dim(\pi_{m-1}(V))$ and $d_m = \dim(\pi_m(V))$. By our aforementioned assumptions, we have that $d_{m-1}, d_m \ge 1$. Now consider the maps $s'_{m-1}, s'_m : [0,1] \to \R_{\ge 0}$ as follows
\begin{align*}
    s'_{m-1}(t) &= \left(s_{m-1} + \frac{d_m}{d_{m-1}} s_m\right)(1-t)\\
    s'_m(t) &= \left(\frac{d_{m-1}}{d_m}s_{m-1} + s_m\right)t.
\end{align*}

Observe that, for all $t \in [0,1]$, we have that
\[
    s'_{m-1}(t)d_{m-1} + s'_m(t)d_m = (d_{m-1}s_{m-1} + d_m s_m)(1-t) + (d_{m-1}s_{m-1} + d_m s_m)t = d_{m-1}s_{m-1} + d_ms_m.
\]
Let $\vec{s}(t) = (s_1, s_2, \hdots, s_{m-2}, s'_{m-1}(t), s'_m(t))$. We have by the above calculation (\ref{eq:dim-cond}) is critical for $\vec{s}(t)$ at $V$ for all $t \in [0,1]$. Further observe at $t_0 = \frac{s_m}{\frac{d_{m-1}}{d_m}s_{m-1} + s_m} \in (0,1)$, we have that $(s'_{m-1}(t_0), s'_m(t_0)) = (s_{m-1}, s_m).$ Thus, we seek to prove (\ref{eq:entropy}) at $\vec{s}(t_0)$.

Let $t_{-} \in [0,1]$ be the minimum $t \in [0,1]$ for which (\ref{eq:dim-cond}) holds at $\vec{s}(t)$ for all $U \subseteq V$. Note that $t_{-} < t_0$ as (\ref{eq:dim-cond}) is not critical at any $U \subsetneq V$. Likewise, we define $t_{+} \in [0,1]$ to be the maximum $t \in [0,1]$ for which (\ref{eq:dim-cond}) holds at $\vec{s}(t)$ for all $U \subseteq V$. We have that $t_{+} > t_0$.

We claim that (\ref{eq:entropy}) holds at $\vec{s}(t_{-})$ and $\vec{s}(t_{+})$. If $t_{-} = 0$, then note that $s'_{m}(t_{-}) = 0$, so $\vec{s}(t_{-})$ has support size at most $m-1$. Therefore (\ref{eq:entropy}) holds at $\vec{s}(t_{-})$ by our induction hypothesis on $m$. Otherwise, $t_{-} > 0$, so there exists nonzero $W \subsetneq V$ for which (\ref{eq:dim-cond}) is critical at $W$ for $\vec{s}(t_{-})$. In that case, (\ref{eq:entropy}) holds at $\vec{s}(t_{-})$ by the argument in Section~\ref{subsec:W}. By similar logic, (\ref{eq:entropy}) holds at $\vec{s}(t_{+})$.

To finish, we seek to prove (\ref{eq:entropy}) at $\vec{s}(t_0)$, but note that since $t_0 \in [t_{-}, t_{+}]$, we have that $\vec{s}(t_0)$ is a convex combination of $\vec{s}(t_{-})$ and $\vec{s}(t_{+})$. Thus, by taking a convex combination of (\ref{eq:entropy}) at $\vec{s}(t_{-})$ and $\vec{s}(t_{+})$, we prove (\ref{eq:entropy}) at $\vec{s}(t_0) = (s_1, \hdots, s_m)$, as desired.

\begin{remark}
We remark the high-level structure of the proof of Theorem~\ref{thm:entropy-bl} (particularly the recursion on a critical subspace) has a striking resemblance to the proofs of the GM-MDS theorem on the generating matrix zero patterns of Reed-Solomon codes~\cite{yildiz2019gmmds,lovett2018gmmds} (see also \cite{BDG24}).
\end{remark}

\subsection*{Acknowledgments} We thank Mahdi Cheraghchi, Sivakanth Gopi, Venkatesan Guruswami, and Madhur Tulsiani for helpful discussions and encouragement.

Joshua Brakensiek  was partially supported by Venkatesan Guruswami's Simons Investigator award and National Science Foundation grants No. CCF-2211972 and No. DMS-2503280. Yeyuan Chen was partially supported by the National Science Foundation grant No. CCF-2236931. Zihan Zhang was partially supported by the National Science Foundation grant No. CCF-2440926.

\subsection*{AI Disclosure}

In Section~\ref{sec:asymptotic}, calculations regarding $W(z)$ were suggested by GPT-5 to wrap the seemingly complicated calculations. We have carefully checked and revised all calculation details. Otherwise, all conceptual and written contributions to this paper were done by the human authors.

\bibliography{main}

\end{document}